\documentclass{birkmult}

\usepackage{latexsym,amsmath,amscd,amssymb,graphics,xypic}
\usepackage{enumerate}

\usepackage{graphicx,lscape}

\usepackage[colorlinks]{hyperref}
\usepackage{url}

\begin{document}

\newtheorem{theorem}{Theorem}[section]
\newtheorem{definition}[theorem]{Definition}
\newtheorem{lemma}[theorem]{Lemma}
\newtheorem{remark}[theorem]{Remark}
\newtheorem{proposition}[theorem]{Proposition}
\newtheorem{corollary}[theorem]{Corollary}
\newtheorem{example}[theorem]{Example}

\numberwithin{equation}{section}
\newcommand{\ep}{\varepsilon}
\newcommand{\R}{{\mathbb  R}}
\newcommand\C{{\mathbb  C}}
\newcommand\Q{{\mathbb Q}}
\newcommand\Z{{\mathbb Z}}
\newcommand{\N}{{\mathbb N}}

\title[On a class of quadratic Hamiltonian systems]{On a class of three dimensional quadratic Hamiltonian systems}

\author{R\u{a}zvan M. Tudoran}

\address{The West University of Timi\c soara,\\
Faculty of Mathematics and Computer Science, \\
Department of Mathematics,\\
Blvd. Vasile P\^ arvan, No. 4,\\
300223 - Timi\c soara,\\
Romania.}

\email{tudoran@math.uvt.ro}



\date{February 22, 2011}
\dedicatory{}

\begin{abstract}
The purpose of this article is to compute the normal form of a class of general quadratic Hamiltonian systems that generalizes naturally Euler's equations from the free rigid body dynamics.
\end{abstract}

\maketitle

\section{Introduction}
\label{section:one}

When thinking about quadratic and homogeneous Hamiltonian systems on the dual of a Lie algebra, the first example that comes to our mind is the system describing the rotations of a free rigid body around its center of mass. This system was derived by Euler in 1758 (see e.g. \cite{euler}) then generalized by Poincar\'e (see \cite{poincare}), and then again later by Arnold (see e.g. \cite{arnoldcarte}) starting from the original $\mathfrak{so}(3)$ Lie algebra, to a general Lie algebra. The mathematical literature contains a huge 
amount of writings concerning Euler's equations, from their original form to the most general forms (see e.g. \cite{euler}, \cite{kowalewski}, \cite{abraham}, \cite{bogo1}, \cite{bogo2}, \cite{holm1}, \cite{holm2}, \cite{holm3}, \cite{ratiu},\cite{ec}, \cite{marsdenratiu}, \cite{rotorpendul}, \cite{tarama}). 

Recently, in \cite{tudoran} it is proved that under suitable non-degeneracy conditions, a quadratic and homogeneous (quadratic) Hamiltonian system defined on the dual of the Lie algebra $\mathfrak{o}(K)$ of real $K$ - skew-symmetric matrices, where $K$ is an arbitrary $3\times 3$ real symmetric matrix, is affinely equivalent to Euler's equations (Euler's equations with three linear controls) describing the free rigid body dynamics. More exactly, in both cases above, the non-degeneracy condition requires the existence of a positive definite linear combination between the quadratic form generated by $K$ and respectively the quadratic form that generates the Hamiltonian of the system. In \cite{tudoran} it is introduced a general family of quadratic Hamiltonian systems that generalizes the quadratic Hamiltonian systems defined on the dual of the Lie algebra $\mathfrak{o}(K)$, and consequently gives also a generalization of Euler's equations describing the free rigid body dynamics (with linear controls). 

The purpose of this paper is to construct explicitly the normal form of an interesting subset of quadratic Hamiltonian system from this general family, that generalize Euler's equations (with linear controls).  

For details on Poisson geometry and Hamiltonian dynamics, see, e.g. \cite{abraham}, \cite{arnoldcarte}, \cite{marsdenratiu}, \cite{holm1}, \cite{holm2}, \cite{holm3}, \cite{ratiurazvan}.

\section{Normal form of a large class of three dimensional quadratic Hamiltonian systems}

In this section we give an explicit method to compute the normal form of a general quadratic Hamiltonian system on a natural extension of $(\mathfrak{o}(K))^*$, where $K\in Sym(3)$ is a given $3\times 3$ real symmetric matrix.

Recall that the Lie algebra $\mathfrak{o}(K)$ is defined as $\mathfrak{o}(K):=\{A\in \mathfrak{gl}(3;\R):\ A^{T}K+KA=O_3\}$. Note that in the case when $K$ is nonsingular, $\mathfrak{o}(K)$ is the Lie algebra of the Lie group $O(K):=\{A\in GL(3;\R):\ A^{T}KA=K \}$ of $K$-orthogonal $3\times 3$ matrices. 

Let us now recall a Lie algebra isomorphism between $\mathfrak{o}(K)$ and $\R^3$. For details regarding this isomorphism, see e.g. \cite{holm1}, \cite{marsdenratiu}.
\begin{proposition}\label{pr.2.1}
The Lie algebras $\left(\mathfrak{o}(K),+,\cdot_{\R},[\cdot,\cdot]\right)$ and $\left(\R^3,+,\cdot_{\R},\times_{K}\right)$ are isomorphic, where $[\cdot,\cdot]$ is the commutator of matrices, and respectively $u\times_{K}v:=K(u\times v)$, for any $u,v\in\mathbb{R}^{3}$.
\end{proposition}
Note that for $K$ nonsingular, the above isomorphism between $S\in\mathfrak{o}(K)$ and $\bold{s}\in\mathbb{R}^3$, can be defined by using the equation $Su=\bold{s}\times Ku$, $u\in\mathbb{R}^3$.

An immediate consequence of the this proposition is that $(\mathfrak{o}(K))^*\cong(\R^3)^*\cong$ $\R^3$, viewed as a dual of a Lie algebra, it has a natural Poisson structure, namely the "minus" Lie-Poisson structure, which in this case proves to be generated by the Poisson bracket:
$$\{f,g\}_{K}:=-\nabla C_{K}\cdot(\nabla f\times\nabla g),$$ 
for any $f,g\in C^\infty(\R^3,\R)$, where the smooth function $C_{K}\in C^\infty(\R^3,\R)$ is given by
$$C_{K}(u):=\dfrac{1}{2}u^{T}Ku.$$
For more details regarding this bracket and the associated Lie-Poisson dynamics, see e.g. \cite{holm1}.  

It is not hard to see that the center of the Poisson algebra $C^\infty(\R^3,\R)$ is generated by the Casimir invariant function
$C_{K}\in C^\infty(\R^3,\R)$, $C_{K}(u)=\dfrac{1}{2}u^TKu$.

Let us now generalize the above setting by introducing the Poisson manifold $(\mathbb{R}^{3},\{\cdot,\cdot\}_{(K,\bold{k})})$, where the Poisson bracket $\{\cdot,\cdot\}_{(K,\bold{k})}$ is defined by:
$$\{f,g\}_{(K,\bold{k})}:=-\nabla C_{(K,\bold{k})}\cdot(\nabla f\times\nabla g),$$ 
for any $f,g\in C^\infty(\R^3,\R)$, and the smooth function $C_{(K,\bold{k})}\in C^\infty(\R^3,\R)$ is given by
$$C_{(K,\bold{k})}(u):=\dfrac{1}{2}u^{T}Ku+u^{T}\bold{k}.$$
Using the above definition of the Poisson bracket, it follows that the smooth function $C_{(K,\bold{k})}$ is a Casimir invariant function.

Consequently, a quadratic Hamiltonian system on $(\mathbb{R}^{3},\{\cdot,\cdot\}_{(K,\bold{k})})$, is generated by a smooth function $H_{(A,\bold{a})}\in C^\infty(\R^3,\R)$, given by $$H_{(A,\bold{a})}(u):=\dfrac{1}{2}u^{T}Au+u^{T}\bold{a},$$
where $A\in Sym(3)$ is an arbitrary real symmetric matrix, and $\bold{a}\in\mathbb{R}^3$. 

Hence, the associated Hamiltonian system is governed by the following differential equation:
\begin{equation}\label{general}
\dot u=(Ku+\bold{k})\times(Au+\bold{a}),\ u\in\R^3.
\end{equation}

The main result in \cite{tudoran} states that if there exists $\alpha,\beta\in\R$ such that $\alpha A + \beta K$ is positive (or negative) definite, then the system \eqref{general} is affinely equivalent to Euler's equations of the free rigid body dynamics with linear controls, namely:
\begin{equation}\label{eulercontrol}
\dot u=u \times(Du+\bold{d}),\ u\in\R^3,
\end{equation}
where $D=\operatorname{diag}(\lambda_1,\lambda_2,\lambda_3)$ is a real diagonal $3\times 3$ matrix, and $\bold{d}\in\R^3$. In the case when $\bold{d}=0$, the equations \eqref{eulercontrol} are exactly the Euler's equations of the free rigid body dynamics.
Note that in coordinates, the system \eqref{eulercontrol} become:
$$\left\{\begin{array}{l}
\dot x_{1}=(\lambda_3-\lambda_2)x_{2}x_{3}+d_{3}x_{2}-d_{2}x_{3},\\
\dot x_{2}=(\lambda_1-\lambda_3)x_{1}x_{3}-d_{3}x_{1}+d_{1}x_{3},\\
\dot x_{3}=(\lambda_2-\lambda_1)x_{1}x_{2}+d_{2}x_{1}-d_{1}x_{2},\\
\end{array}\right.
$$
where $(d_1,d_2,d_3)$ are the coordinates of $\bold{d}$.

The purpose of this paper is to give a normal form of the equations \eqref{general} in the case when the matrices $A$ and $K$ commutes. Let us now state the main result of this paper.

\begin{theorem}\label{mr}
If the matrices $A$ and $K$ commutes, then the system \eqref{general} is orthogonally equivalent to
the dynamical system:
\begin{equation}\label{sok}
\dot v=(D_{K}v+\hat{\bold{k}})\times(D_{A}v+\hat{\bold{a}}),\ v\in\R^3,
\end{equation}
where the real diagonal $3\times 3$ matrices $D_A$, $D_K$ are given by $D_A =R^{T}AR$, $D_K =R^{T}KR$, where $R\in O(Id)=O(3,\R)$ is a $3\times 3$ orthogonal matrix, and $\hat{\bold{k}}:=\det(R)R^{T}\bold{k}$, $\hat{\bold{a}}:=\det(R)R^{T}\bold{a}$.
\end{theorem}
\begin{proof}
As a consequence of the commutativity of the real symmetric matrices $A$ and $K$ (i.e., $[A,K]=AK-KA=0$), there exists an orthogonal matrix $R\in O(Id)=O(3,\R)$ which diagonalize simultaneously both the matrices $A$ and $K$. More exactly, there exists two $3\times 3$ real diagonal matrices $D_A$ and respectively $D_K$ such that $D_A =R^{T}AR$, and respectively $D_K =R^{T}KR$. 

Let us denote $R^{-T}:=(R^{-1})^{T}$. Recall that for any matrix $R\in O(Id)=O(3,\R)$ we have that $(\operatorname{det}(R))^{2}=1$.

Next, we will show that:
\begin{equation*}
u(t)=\operatorname{det}(R)Rv(t),
\end{equation*}
where $t\mapsto u(t)$ is a solution of the dynamical system
\eqref{general}:
$$\dot u=(Ku+\bold{k})\times(Au+\bold{a}),$$
and respectively $t\mapsto v(t)$, is a solution of the dynamical
system
$$\dot v=(D_K v+\hat{\bold{k}})\times(D_A v+\hat{\bold{a}}).$$
To prove this assertion, note first that:
\begin{align*}
(Ku+\bold{k})&\times (Au+\bold{a})=(K\operatorname{det} (R)Rv+\bold{k})\times (A\operatorname{det} (R)Rv+\bold{a})\\
&=[\operatorname{det}(R)KRv+(\operatorname{det} (R))^{2}\bold{k}]\times [\operatorname{det} (R)ARv+(\operatorname{det}(R))^{2}\bold{a}]\\
&=(\operatorname{det} (R))^{2}[(KRv+\operatorname{det}(R)\bold{k})\times (ARv+\operatorname{det}(R)\bold{a})]\\
&=(KRv+\operatorname{det}(R)\bold{k})\times (ARv+\operatorname{det}(R)\bold{a})\\
&=(RD_K v+\operatorname{det}(R)\bold{k})\times (RD_A v+\operatorname{det}(R)\bold{a})\\
&=(RD_K v+\operatorname{det}(R)RR^{T}\bold{k})\times (RD_A v+\operatorname{det}(R)RR^{T}\bold{a})\\
&=[R(D_K v+\operatorname{det}(R)R^{T}\bold{k})]\times [R(D_A v+\operatorname{det}(R)R^{T}\bold{a})]\\
&=\operatorname{det} (R) R^{-T}[(D_K v+\hat{\bold{k}})\times (D_A v+\hat{\bold{a}})]\\
&=\operatorname{det} (R) R[(D_K v+\hat{\bold{k}})\times (D_A v+\hat{\bold{a}})].
\end{align*}
Hence,
\begin{align*}
\dot u=(Ku+\bold{k})\times (Au+\bold{a})&\Leftrightarrow \operatorname{det}(R)R\dot v=\operatorname{det} (R) R[(D_K v+\hat{\bold{k}})\times (D_A v+\hat{\bold{a}})]\\
&\Leftrightarrow\dot v=(D_K v+\hat{\bold{k}})\times (D_A v+\hat{\bold{a}}).
\end{align*}
\end{proof}
\begin{remark}
Using coordinates, the system \eqref{sok} becomes:
$$\left\{\begin{array}{l}
\dot x_1=(K_2 A_3-K_3 A_2)x_2 x_3+(K_2 a_3-k_3 A_2)x_2 +(k_2 A_3-K_3 a_2)x_3+k_2 a_3-k_3 a_2,\\
\dot x_2=(K_3 A_1-K_1 A_3)x_1 x_3+(k_3 A_1-K_1 a_3)x_1 +(K_3 a_1-k_1 A_3)x_3+k_3 a_1-k_1 a_3,\\
\dot x_3=(K_1 A_2-K_2 A_1)x_1 x_2+(K_1 a_2-k_2 A_1)x_1 +(k_1 A_2-K_2 a_1)x_2+k_1 a_2-k_2 a_1,\\
\end{array}\right.
$$
where $D_A=\operatorname{diag}(A_1,A_2,A_3)$, $D_K=\operatorname{diag}(K_1,K_2,K_3)$, $\hat{\bold{a}}=(a_1,a_2,a_3)$, and $\hat{\bold{k}}=(k_1,k_2,k_3)$.
\end{remark}

\subsection*{Acknowledgment}
The author is supported by CNCSIS - UEFISCDI, project number PN II - IDEI code 1081/2008 No. 550/2009.

\end{document}